\documentclass[a4paper,twoside, openright, 12pt, leqno]{article}

\usepackage[top=3.5cm, bottom=2.5cm, outer=2.5cm, inner=2.5cm, headsep=1.5cm, footskip=1.2cm, headheight=14pt]{geometry}
\usepackage{amsmath}
\usepackage{amsthm}                             % Proof environment
\usepackage{amssymb}
\usepackage{enumitem}				% Different enumarate options
\usepackage{graphicx}				% standard LaTeX graphics tool
\usepackage{natbib}				% Bibliography
\usepackage{setspace}				% Control line space
\usepackage{afterpage}				% Control the subsequent page
\usepackage[bottom]{footmisc}			% Manage footnotes. Used to modify the indentation and force footnotes bottom page. Necessary befor hyperref !!!
\usepackage{etoolbox}				% Add patches
\usepackage{url}				% Include URL in the references
\usepackage[width=0.95\textwidth]{caption} 	% Mange font size and margins of captions
\usepackage{changepage}				% Temporary change page settings
\usepackage[rightmargin=0em]{quoting}		% Add quotations in text
\usepackage{multirow}				% Multiple rows in tables
\usepackage{placeins}				% Use \FloatBarrier to put figures and tables in current section
\usepackage{booktabs}
\usepackage[hyperfootnotes=false, colorlinks = true, allcolors = blue]{hyperref}		% Include hyperlinks
\usepackage{fancyhdr}				% Fancy headers

\makeatletter
\patchcmd{\NAT@citex}
{\@citea\NAT@hyper@{%
		\NAT@nmfmt{\NAT@nm}%
		\hyper@natlinkbreak{\NAT@aysep\NAT@spacechar}{\@citeb\@extra@b@citeb}%
		\NAT@date}}
{\@citea\NAT@nmfmt{\NAT@nm}%
	\NAT@aysep\NAT@spacechar\NAT@hyper@{\NAT@date}}{}{}
\patchcmd{\NAT@citex}
{\@citea\NAT@hyper@{%
		\NAT@nmfmt{\NAT@nm}%
		\hyper@natlinkbreak{\NAT@spacechar\NAT@@open\if*#1*\else#1\NAT@spacechar\fi}%
		{\@citeb\@extra@b@citeb}%
		\NAT@date}}
{\@citea\NAT@nmfmt{\NAT@nm}%
	\NAT@spacechar\NAT@@open\if*#1*\else#1\NAT@spacechar\fi\NAT@hyper@{\NAT@date}}
{}{}
\makeatother

\usepackage{array}
\newcolumntype{L}[1]{>{\raggedright\let\newline\\\arraybackslash\hspace{0pt}}m{#1}}
\newcolumntype{C}[1]{>{\centering\let\newline\\\arraybackslash\hspace{0pt}}m{#1}}
\newcolumntype{R}[1]{>{\raggedleft\let\newline\\\arraybackslash\hspace{0pt}}m{#1}}

\newcommand*\xbar[1]{%
   \hbox{%
     \vbox{%
       \hrule height 0.5pt % The actual bar
       \kern0.3ex%         % Distance between bar and symbol
       \hbox{%
         \kern-0.25em%      % Shortening on the left side
         \ensuremath{#1}%
         \kern-0.05em%      % Shortening on the right side
       }%
     }%
   }%
} 

\newtheorem{theorem}{Proposition}

\setcitestyle{aysep={}}	

\newcommand\fnsep{\textsuperscript{,}}

\makeatletter
\def\@xfootnote[#1]{%
  \protected@xdef\@thefnmark{#1}%
  \@footnotemark\@footnotetext}
\makeatother

\begin{document}

% Affilations as footnotes with symmbols
\renewcommand{\thefootnote}{\alph{footnote}}

% Remove header from the first page
\thispagestyle{empty}

\begin{center}
    
    \vspace*{1cm}
    \LARGE{The threshold age of Keyfitz' entropy}	
    \vspace{.4cm}

    \vspace{1cm}
    \large Jos\'e Manuel Aburto\footnote[*]{Correspondence:~\href{mailto:jmaburto@sdu.dk}{jmaburto@sdu.dk}.}\footnote{Interdisciplinary Center on Population Dynamics, University of Southern Denmark, Odense, Denmark.}\fnsep\footnote{Max Planck Institute for Demographic Research, Rostock, Germany.}, Jesus-Adrian Alvarez\textsuperscript{a}, Francisco Villavicencio\footnote{Institute for International Programs, Bloomberg School of Public Health, Johns Hopkins University, Baltimore, MD, United States.}\fnsep\textsuperscript{a}\\ and James W. Vaupel\textsuperscript{a}\fnsep\textsuperscript{b}
    
    \vspace{1cm}
    \large\today
    \vspace{1cm}
       
\end{center}

% In text use numbers for footnotes
\renewcommand{\thefootnote}{\arabic{footnote}}
\setcounter{footnote}{0}

\section*{Abstract}
\bigskip

\textbf{BACKGROUND} \\
Indicators of relative inequality of lifespans are important because they capture the dimensionless shape of aging. They are markers of inequality at the population level and express the uncertainty at the time of death at the individual level. In particular, Keyfitz' entropy $\xbar{H}$ represents the elasticity of life expectancy to a change in mortality and it has been used as an indicator of lifespan variation. However, it is unknown how this measure changes over time and whether a threshold age exists, as it does for other lifespan variation indicators.
\bigskip

\noindent\textbf{RESULTS} \\
The time derivative of $\xbar{H}$ can be decomposed into changes in life disparity $e^\dagger$ and life expectancy at birth $e_o$. Likewise, changes over time in $\xbar{H}$ are a weighted average of age-specific rates of mortality improvements. These weights reflect the sensitivity of $\xbar{H}$ and show how mortality improvements can increase (or decrease) the relative inequality of lifespans. Further, we prove that $\xbar{H}$, as well as $e^\dagger$, in the case that mortality is reduced in every age, has a threshold age below which saving lives reduces entropy, whereas improvements above that age increase entropy.
\bigskip

\noindent\textbf{CONTRIBUTION} \\
We give a formal expression for changes over time of $\xbar{H}$ and provide a formal proof of the threshold age that separates reductions and increases in lifespan inequality from age-specific mortality improvements.

% Line interpsace
\linespread{1.5}\normalsize
\clearpage

%%%%%%%%%%%%%%%%%%%%%%%%%%%%%%%%%%%%%%%%%%%%%%%%%%%%%%%%%%%%%%%%%%%%%%%%%%%%%%%%%%%%%%%%
%%%%%%%%%%%%%%%%%%%%%%%%%%%%%%%%%%%%%%%%%%%%%%%%%%%%%%%%%%%%%%%%%%%%%%%%%%%%%%%%%%%%%%%% 

\section{Relationship}
Keyfitz' entropy is a dimensionless indicator of the relative variation in the length of life compared to life expectancy \citep{Keyfitz1977, demetrius1978adaptive}. It is usually defined as
\begin{equation*}
\xbar{H}(t)=-\frac{\int_0^\infty\ell(a, t)\,\ln\ell(a, t)\,da}{\int_0^\infty\ell(a, t)\,da}=\int_0^\infty c(a, t)\,H(a,t)\,da=\frac{e^\dagger(t)}{e_o(t)}\;,
% \label{eq:entropy}
\end{equation*}
where $e^\dagger(t)=-\int_0^\infty\ell(a,t)\,\ln\ell(a,t)\,da$ is the life disparity or number of life-years lost as a result of death \citep{Vaupel2003}, $e_o(t)=\int_0^\infty\ell(a, t)\,da$ is the life expectancy at birth at time $t$, $\ell(a,t)$ is the life table survival function, $c(a,t)=\ell(a)\,/\,\int_0^\infty\ell(x)\,dx$ is the population structure, and $H(a,t)=\int_0^a\mu(x,t)\,dx$ is the cumulative hazard to age $a$, where $\mu(x,t)$ is the force of mortality (hazard rate or risk of death) at age $x$ at time $t$. Note that $\xbar{H}(t)$ can be interpreted as an average value of $H(a,t)$ in the population at time $t$.

\cite{Goldman1986} and \cite{Vaupel1986} proved that
$$
e^\dagger(t)=\int_0^\infty d(a, t)\,e(a, t)\,da\;,
$$
where $d(a,t)$ represents the distribution of deaths and $e(a,t)=\int_a^\infty\ell(x,t)\,dx\,/\,\ell(a,t)$ the remaining life expectancy at age $a$. This provides an alternative expression for Keyfitz' entropy:
$$
\xbar{H}(t)=\frac{\int_0^\infty d(a, t)\,e(a, t)\,da}{\int_0^\infty\ell(a, t)\,da}\;.
$$

Let $\dot{\xbar{H}}$ denote the partial derivative of $\xbar{H}$ with respect to time.\footnote{In the following, a dot over a function will denote its partial derivative with respect to time $t$, but variable $t$ will be omitted for simplicity.} We define $\rho(x)=-\dot{\mu}(x)\,/\,\mu(x)$ as the age-specific rates of mortality improvements. Then, the relative derivative of $\xbar{H}$ can be expressed as a weighted average of age-specific rates of mortality improvement, 
\begin{equation}
\dot{\xbar{H}}\,/\,\xbar{H} = \int_0^\infty\rho(x)\,w(x)\,W(x)\,dx\;,
\label{eq:1}
\end{equation}
with weights
$$
w(x)=\mu(x)\,\ell(x)\,e(x)\qquad\text{and}\qquad W(x)=\frac{1}{e^\dagger}\,\big(H(x)+\xbar{H}(x)-1\big)-\frac{1}{e_o}\;.
$$

Function $\xbar{H}(x)$ is Keyfitz' entropy conditioned on surviving to age $x$, defined as
$$
\xbar{H}(x)=\frac{e^\dagger(x)}{e(x)}=\frac{\int_x^\infty d(a)\,e(a)\,da}{\int_x^\infty\ell(a)\,da}\,.
$$
where $e^\dagger(x)=\int_x^\infty d(a)\,e(a)\,da\,/\,\ell(x)$ refers to life disparity above age $x$, and $e(x)$ is the remaining life expectancy at age $x$.

Note that Keyfitz' entropy $\xbar{H}$ is a measure of lifespan inequality. Thus, higher values represent more lifespan disparity, whereas lower values denote less variation of lifespans. If mortality improvements over time occur at all ages, there exists a unique threshold age $a^H$ that separates \emph{positive} from \emph{negative} contributions to Keyfitz' entropy $\xbar{H}$ resulting from those mortality improvements. This threshold age $a^H$ is reached when
\begin{equation}
  H\left(a^H\right)+ \xbar{H}\left(a^H\right) =1 + \xbar{H}\;.
  \label{threhsold.age}
\end{equation}

%%%%%%%%%%%%%%%%%%%%%%%%%%%%%%%%%%%%%%%%%%%%%%%%%%%%%%%%%%%%%%%%%%%%%%%%%%%%%%%%%%%%%%%%
%%%%%%%%%%%%%%%%%%%%%%%%%%%%%%%%%%%%%%%%%%%%%%%%%%%%%%%%%%%%%%%%%%%%%%%%%%%%%%%%%%%%%%%%

\section{Proof}

\cite{Fernandez2015} showed that the relative derivative of  $\xbar{H}$ can be expressed as
\begin{equation}
\dot{\xbar{H}}\,/\,\xbar{H}=\frac{\dot{e}^\dagger}{e^\dagger}-\frac{\dot{e}_o}{e_o}\;,
\label{eq:alternative}
\end{equation}
This formula indicates that relative changes in $\xbar{H}$ over time are given by the difference between relative changes in $e^\dagger$ (dispersion component) and relative changes in $e_o$ (translation component). We will first provide expressions for $\dot{e}_o$ and $\dot{e}^\dagger$ to prove that~\eqref{eq:1} and~\eqref{eq:alternative} are equivalent. Next, we will prove the existence of threshold age for $\xbar{H}$ and its uniqueness.

\subsection{Relative changes over time in $\xbar{H}$}

\cite{Vaupel2003} showed that changes over time in life expectancy at birth are a weighted average of the total rates of mortality improvements:
\begin{equation}
\label{ex.derivative}
\dot{e}_o=\int_0^\infty\rho(x)\,w(x)\,dx\;,
\end{equation}
%.
where $\rho(x)=-\dot{\mu}(x)\,/\,\mu(x)$ are the age-specific rates of mortality improvement, and $w(x)=\mu(x)\,\ell(x)\,e(x) = d(x)e(x)$ is a measure of the importance of death at age $x$. 

Since $d(x)=\mu(x)\,\ell(x)$ and $\ell(x)\,e(x)=\int_x^\infty\ell(a)\,da$, the partial derivative with respect to time of $e^\dagger=\int_0^\infty d(a)\,e(a)\,da$ can be expressed as
\begin{equation*}
 \begin{split}
 \dot{e}^\dagger	& = \int_0^\infty\dot{\mu}(x)\,\ell(x)\,e(x)\,dx+\int_0^\infty\mu(x)\int_x^\infty\dot{\ell}(a)\,da\,dx		\\
			& = -\int_0^\infty\rho(x)\,w(x)\,dx+\int_0^\infty\dot{\ell}(a)\int_0^a\mu(x)\,dx\,da                         \\
			& = -\int_0^\infty\rho(x)\,w(x)\,dx+\int_0^\infty\dot{\ell}(a)\,H(a)\,da     \\
			& = -\int_0^\infty\rho(x)\,w(x)\,dx-\int_0^\infty\int_0^a\dot{\mu}(x)\,dx\,\ell(a)\,H(a)\,da\;,
 \end{split}
\end{equation*}
where $H(a)$ is the cumulative hazard to age $a$. By reversing the order of integration and doing some additional manipulations, we get

\begin{equation}
 \begin{split}
 \dot{e}^\dagger	
    & = -\int_0^\infty\rho(x)\,w(x)\,dx-\int_0^\infty\dot{\mu}(x)\int_x^\infty\,\ell(a)\,H(a)\,da\,dx    \\
    & = -\int_0^\infty\rho(x)\,w(x)\,dx+\int_0^\infty\rho(x)\,w(x)\,\frac{\int_x^\infty\,\ell(a)\,H(a)\,da}{\ell(x)\,e(x)}\,dx                                   \\
    & =\int_0^\infty\rho(x)\,w(x)\left(\frac{\int_x^\infty\,\ell(a)\big(H(a)-H(x)+H(x)\big)\,da}{\ell(x)\,e(x)}-1\right)dx                  \\
    & =\int_0^\infty\rho(x)\,w(x)\left(H(x)\,\frac{\int_x^\infty\,\ell(a)\,da}{\ell(x)\,e(x)}+\frac{\int_x^\infty\,\ell(a)\big(H(a)-H(x)\big)\,da}{\ell(x)\,e(x)}-1\right)\,dx                                           \\
    & =\int_0^\infty\rho(x)\,w(x)\left(H(x)+\frac{\int_x^\infty\,\ell(a)\big(H(a)-H(x)\big)\,da}{\ell(x)\,e(x)}-1\right)\,dx\;.
 \label{der.edagger}
 \end{split}
\end{equation}

In Proposition~\ref{prop1} in the Appendix, we prove that
\begin{equation}
e^\dagger(x)=\frac{1}{\ell(x)}\int_x^\infty d(a)\,e(a)\,da=\frac{1}{\ell(x)}\int_x^\infty\,\ell(a)\big(H(a)-H(x)\big)\,da\;.
\label{eq:edaggerNew}
\end{equation}
Replacing~\eqref{eq:edaggerNew} into~\eqref{der.edagger} yields
\begin{equation}
  \begin{split}
    \dot{e}^\dagger
        & = \int_0^\infty\rho(x)\,w(x)\left(H(x)+\frac{e^\dagger(x)}{e(x)}-1\right)dx                 \\
        & = \int_0^\infty\rho(x)\,w(x)\big(H(x)+\xbar{H}(x)-1\big)\,dx   \;.
  \end{split}
  \label{der3.edagger}
\end{equation}
Finally, replacing the expressions of $\dot{e}_o$ and $\dot{e}^\dagger$ from~\eqref{ex.derivative} and~\eqref{der3.edagger} into~\eqref{eq:alternative}, we get 
\begin{equation*}
 \begin{split}
    \dot{\xbar{H}}\,/\,\xbar{H}
        & = \frac{1}{e^\dagger}\int_0^\infty\rho(x)\,w(x)\left(H(x)+\xbar{H}(x)-1\right)dx-\frac{1}{e_o}\int_0^\infty\rho(x)\,w(x)\,dx    \\
        & = \int_0^\infty\rho(x)\,w(x)\left(\frac{1}{e^\dagger}\left(H(x)+\xbar{H}(x)-1\right)-\frac{1}{e_o}\right)dx                               \\
        & = \int_0^\infty\rho(x)\,w(x)\,W(x)\,dx\;,
 \end{split}
\end{equation*}
which proves~\eqref{eq:1} and shows that relative changes over time in Keyfitz' entropy are the average of the rates of mortality improvement weighted by the product $w(x)\,W(x)$.\hfill$\square$

\subsection{The threshold age for $\xbar{H}$}

Using~\eqref{eq:1}, changes over time in Keyfitz' entropy $\xbar{H}$ are given by the function
\begin{equation}
  \dot{\xbar{H}}=\xbar{H}\,\int_0^\infty\rho(x)\,w(x)\,W(x)\,dx\;.
  \label{eq:Hprime}
\end{equation}
If $\dot{\xbar{H}}>0$, lifespan inequality increases over time, whereas $\dot{\xbar{H}}<0$ implies that variation of lifespans decrease over time. Because $\ell(x)$ is a positive function bounded between 0 and 1, Keyfitz' entropy $\xbar{H}>0$. Moreover, assuming age-specific death rates $\mu(x)$ improve over time for all ages, then $\dot{\mu}(x)<0$ and $\rho(x)>0$ at any age $x$. Therefore,~\eqref{eq:Hprime} implies that
\begin{enumerate}
  \item Those ages $x$ in which $w(x)\,W(x)>0$ will contribute \emph{positively} to Keyfitz' entropy $\xbar{H}$ and increase lifespan variation;
  \item Those ages $x$ in which $w(x)\,W(x)<0$ will contribute \emph{negatively} to Keyfitz' entropy $\xbar{H}$ and favor lifespan equality;
  \item Those ages $x$ in which $w(x)\,W(x)=0$ will have no effect on the variation over time of $\xbar{H}$.
\end{enumerate}
Our goal is to prove that if mortality improvements occur for all ages and $\rho(x)>0$, there exists a unique threshold age $a^H$ such that $w\left(a^H\right)\,W\left(a^H\right)=0$. That threshold age will separate \emph{positive} from \emph{negative} contributions to $\xbar{H}$ resulting from mortality improvements. 

The product $w(x)\,W(x)$ can be re-expressed as 
\begin{equation*}
  \begin{split}
    w(x)\,W(x)	& = \mu(x)\,\ell(x)\,e(x)\,\left(\,\frac{1}{e^\dagger}\left(H(x)+\xbar{H}(x)-1\right)-\frac{1}{e_o}\right)	                            \\
                & = \frac{\mu(x)\,\ell(x)\,e(x)}{e^\dagger}\left(H(x)+\xbar{H}(x)-\xbar{H}-1\right)\;.
  \end{split}
  \label{W_2}
\end{equation*}
Since $\mu(x)$, $\ell(x)$, $e(x)$ and $e^\dagger$ are all positive functions, the threshold age of $\xbar{H}$ occurs when
\begin{equation}
\label{g_x}
g(x)= H(x)+\xbar{H}(x)-\xbar{H}-1=0\;.
\end{equation}

When $x$ is close to 0, $g(x)$ takes negative values since
\begin{equation*}
g(0)=H(0)+\xbar{H}^+(0)-\xbar{H}-1=0+\xbar{H}-\xbar{H}-1=-1<0\;.
\end{equation*}
Likewise, $g(x)$ takes positive values when $x$ becomes arbitrary large. Note that $\xbar{H}$ does not depend on age, and therefore
$$
\lim_{x\to\infty}g(x)=\lim_{x\to\infty}\left(H(x)+\xbar{H}(x)\right)=\infty\;,
$$
because $\lim_{x\to\infty}H(x)=\infty$. By definition, $\xbar{H}(x)\geq0$ for all $x$, so regardless of the behavior of $\xbar{H}(x)$ when $x$ is arbitrarily large, the limit of $g(x)$ tends to infinity. Hence, given that $g(0)=-1$ and $\lim_{x\to\infty}g(x)=\infty$, in a continuous framework the intermediate value theorem guarantees the existence of at least one age $a^H$ at which $g(a^H)=0$.

Moreover, as shown in Propostion~\ref{prop2} in the Appendix, $g(x)$ is an increasing function. Therefore, there is a unique threshold age $a^H$ that separates \emph{positive} from \emph{negative} contributions to Keyfitz' entropy $\xbar{H}$, and that threshold age is reached when
\begin{equation*}
w(x)\,W(x)=0\Longleftrightarrow g(x)=0\Longleftrightarrow H(x)+ \xbar{H}(x)= 1+\xbar{H}\;,
\end{equation*}
which proves~\eqref{threhsold.age}.\hfill$\square$

%%%%%%%%%%%%%%%%%%%%%%%%%%%%%%%%%%%%%%%%%%%%%%%%%%%%%%%%%%%%%%%%%%%%%%%%%%%%%%%%%%%%%%%%
%%%%%%%%%%%%%%%%%%%%%%%%%%%%%%%%%%%%%%%%%%%%%%%%%%%%%%%%%%%%%%%%%%%%%%%%%%%%%%%%%%%%%%%%

\section{Related results}

Demographers have developed a battery of indicators to measure how lifespans vary in populations \citep{colchero2016emergence,vanRaalte2013}. The most used indexes are the variance \citep{edwards2005inequality, tuljapurkar2011variance}, standard deviation \citep{vanraalte2018Science}, or coefficient of variation \citep{aburto2018potential} of the age at death distribution , the Gini coefficient \citep{Shkolnikov2003,archer2018diet,gigliarano2017longevity}, Theil index \citep{Smits2009} and years of life lost \citep{Vaupel2011, Aburto2018Eastern} among others. However, only few studies have analytically derived formulas for \textit{threshold age} below and above which mortality improvements respectively decrease and increase lifespan variation. \citet{Zhang2009} showed that the threshold age $(a^\dagger)$ for life disparity $(e^\dagger)$ occurs when $H(x)+ \bar{H}(x)=1$. Similarly, \citet{Gillespie2014} determined a threshold age for the variance of the age at death distribution. \Citet{vanRaalte2013} also showed that it is possible to determine the threshold age by performing an empirical sensitivity analysis of lifespan variation indicators.

 In this article, we contribute to the lifespan variation literature by deriving the threshold age $a^H$ for Keyfitz' entropy. This age separates negative from positive contributions of age-specific mortality improvements. We analytically proved its existence and demonstrated in Section \ref{sec:application} that it differs from the threshold age of $e^\dagger$.

%%%%%%%%%%%%%%%%%%%%%%%%%%%%%%%%%%%%%%%%%%%%%%%%%%%%%%%%%%%%%%%%%%%%%%%%%%%%%%%%%%%%%%%%
%%%%%%%%%%%%%%%%%%%%%%%%%%%%%%%%%%%%%%%%%%%%%%%%%%%%%%%%%%%%%%%%%%%%%%%%%%%%%%%%%%%%%%%%

\section{Applications}
\label{sec:application}

\subsection{Numerical findings}

Figure~\ref{fig:Fig1} depicts the threshold ages for the two related measures: $e^\dagger$ and $\xbar{H}$. Calculations were performed using data from the \citet{HMD} for females in the United States and Italy in 2005. The blue line represents $g(x)$ from Equation \eqref{g_x}. The threshold age $a^H$ occurs when  $g(x)$ crosses zero. The red and grey line display the same functions that \cite{Zhang2009} used to find the threshold age for $e^\dagger$. The intersection of these two lines denotes the threshold age $a^\dagger$. Finally, the dashed black line depicts the life expectancy at birth. \citet{Vaupel2011} noted that $a^\dagger$ tends to fall just below life expectancy. The threshold age for Keyfitz' entropy $a^H$ is greater than $a^\dagger$ and is very close above life expectancy for these countries. Not the similarity of the formulas for $a^\dagger$ given by $H(a^\dagger)+ \bar{H}(a^\dagger)=1$ and $a^H$ given by $H(a^H)+ \xbar{H}(a^H)= 1+\xbar{H}$.

\begin{figure}[h]
  \centering
  \includegraphics[scale=.72]{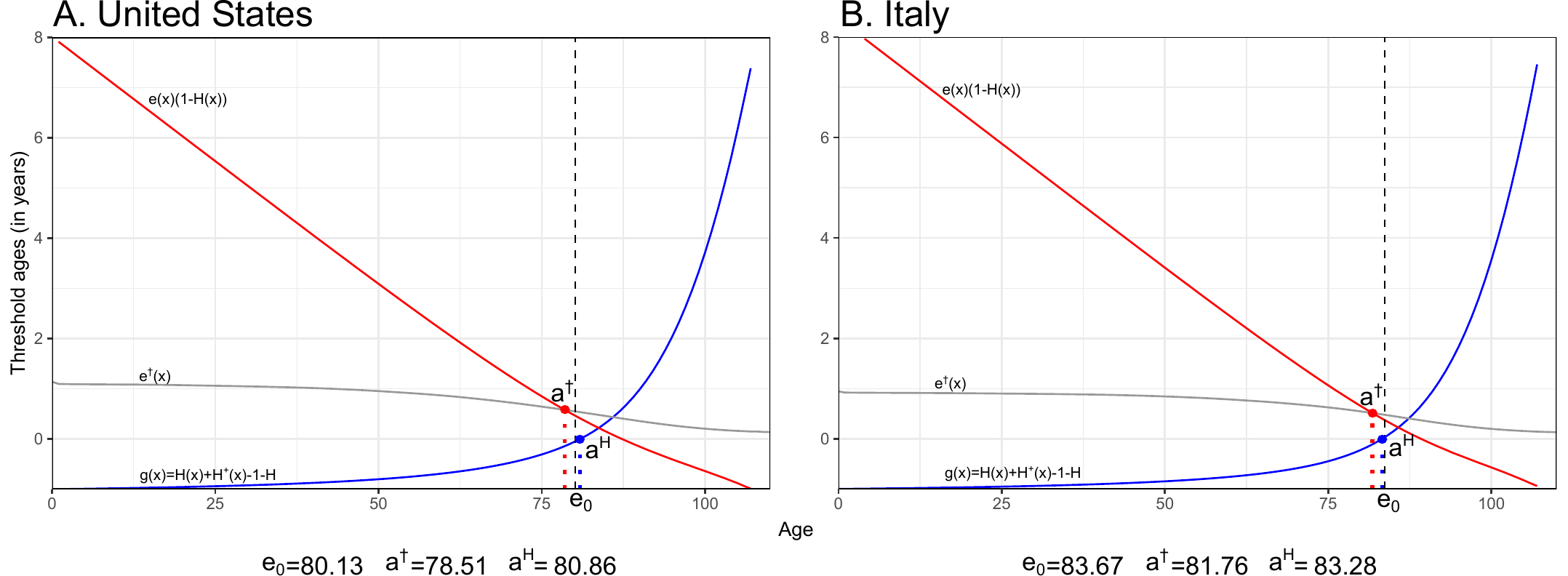}
  \caption{Derivation of threshold ages for $e^\dagger$ ($a^\dagger$) and Keyfitz' entropy ($a^H$) from life table functions for United States and Italy in 2005. Source: \cite{HMD}}
  \label{fig:Fig1}
\end{figure}

Panel A and B of Figure~\ref{fig:Fig2} illustrate the evolution of the threshold ages for $e^\dagger$ and $\xbar{H}$ for females in France and Sweden respectively. We chose these countries because they portray large series of reliable data available through the \cite{HMD}. 

Values for $a^\dagger$ are close to life expectancy throughout the period. However, after around 1950 there is a crossover between $a^\dagger$ and $e_o$ such that $a^\dagger$ remained close to life expectancy but below it. This result shows that the threshold age $a^\dagger$ being below life expectancy is a modern feature of ageing populations with high life expectancy. From the beginning of the period of observation to the 1950s, the threshold age for Keyfitz' entropy was above life expectancy for both countries. During some periods $a^\dagger$ was roughly constant whereas life expectancy trended upwards. After the 1950s, $a^\dagger$ converged towards life expectancy. The code and data to reproduce these results are publicly available through this repository [link not given to avoid identification of authors].

\begin{figure}[h]
  \centering
  \includegraphics[scale=.72]{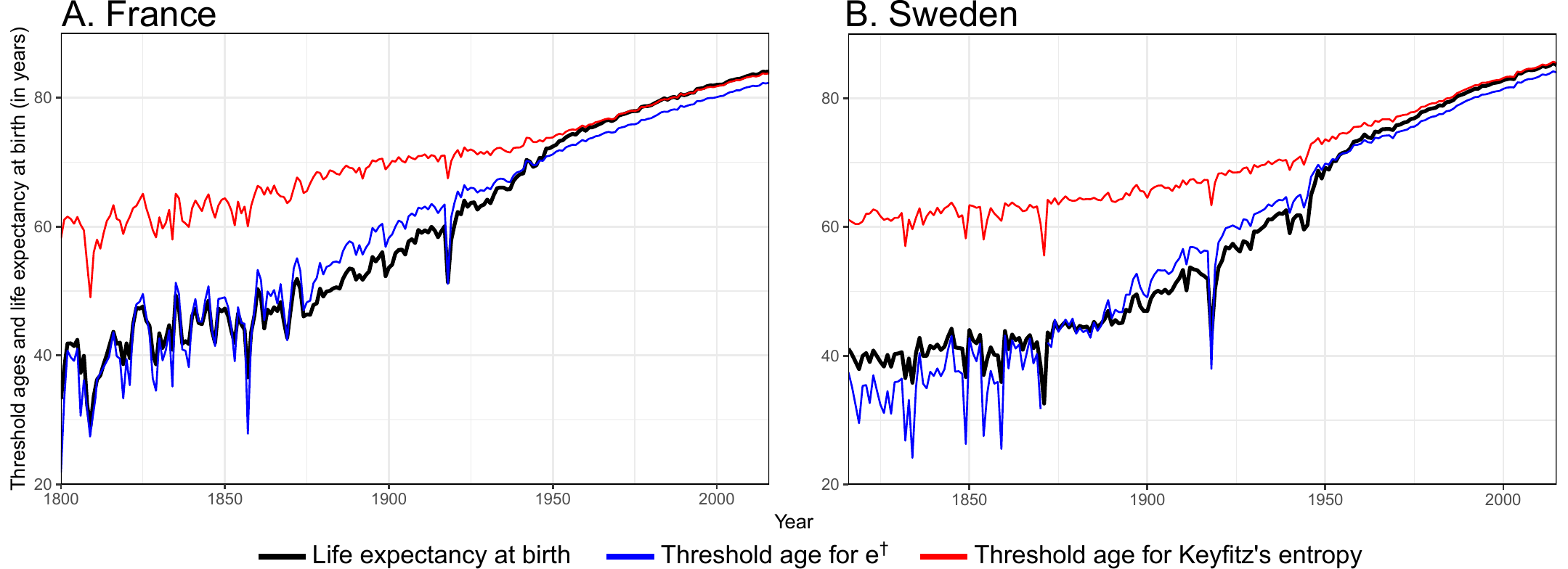}
  \caption{Threshold age for $e^\dagger$ and Keyftiz' entropy $\xbar{H}$ for females in France and Sweden over time. Source: \cite{HMD}}
  \label{fig:Fig2}
\end{figure}

\FloatBarrier

\subsection{Decomposition of the relative derivative of $\xbar{H}$}

The relative derivative of $\xbar{H}$ defined in Equation~\eqref{eq:1} can be decomposed between components before and after the threshold age $a^H$ as follows:
\begin{equation}
 \begin{split}
 \dot{\xbar{H}}\,/\,\xbar{H} &= \int_0^\infty\rho(x)\,w(x)\,W(x)\,dx\;\\
  & = \int_0^{a^H} \rho(x)\,w(x)\,W(x)\,dx + \int_{a^H}^\infty\rho(x)\,w(x)\,W(x)\,dx		 \\
		& = \underbrace{ \left\lbrace \frac{\dot{e}^\dagger[x|x < a^H]}{e^\dagger}- \frac{\dot{e}_o[x|x < a^H]}{e_o}  \right\rbrace }_{\textit{Early life component}} + \underbrace{\left\lbrace \frac{\dot{e}^\dagger[x|x > a^H]}{e^\dagger}-\frac{\dot{e}_o[x|x > a^H]}{e_o} \right\rbrace }_{\textit{Late life component}}
 \end{split}
 \label{Components2}
\end{equation}

If mortality reductions occur at every age, the \textit{early life component} in Equation~\eqref{Components2} is always positive (contributing to reduce entropy) while the \textit{late life component} is negative (contributing to increasing entropy). Thus, it is clear that a negative relationship between life expectancy and entropy over time occurs if the early life component outpaces the late life component. This decomposition is based on the additive properties of the derivatives of life expectancy and $e^\dagger$ as previously shown in \citet{Vaupel2003} and \citet{Fernandez2015}.

%%%%%%%%%%%%%%%%%%%%%%%%%%%%%%%%%%%%%%%%%%%%%%%%%%%%%%%%%%%%%%%%%%%%%%%%%%%%%%%%%%%%%%%%
%%%%%%%%%%%%%%%%%%%%%%%%%%%%%%%%%%%%%%%%%%%%%%%%%%%%%%%%%%%%%%%%%%%%%%%%%%%%%%%%%%%%%%%%

\section{Conclusion}
 Several authors have been interested in decomposing changes over time in life expectancy \citep{arriaga1984measuring, Vaupel1986, pollard1988decomposition, Vaupel2003, beltran2008integrated, beltran2011unifying}. Recently, authors have also investigated how life disparity fluctuations over time can be decomposed \citep{Wagner2010, Zhang2009, Shkolnikov2011, Aburto2018Eastern}. In this paper, we bring both perspectives together and shed light on the dynamics behind changes in Keyfitz' entropy.
   
 \cite{Keyfitz1977} first proposed $\xbar{H}$ as a life table function that measures \textit{the change in life expectancy at birth consequent on a proportional change in age-specific rates}. Since then, several authors have been interested in this measure and its use \citep{demetrius1978adaptive, Demetrius1979,mitra1978short,Goldman1986,Vaupel1986,Hakkert1987,hill1993entropy,Fernandez2015}. Keyfitz' entropy has become an appropriate indicator of lifespan variation that permits comparison of the shape of ageing across different species and over time \citep{baudisch2013pace,Wrycza2015}. In this article, we uncover the mathematical regularities behind the changes over time in Keyfitz' entropy. In particular, this study contributes to the existing literature by showing that (1) Keyfitz' entropy can be decomposed as a weighted average of rates of mortality improvements and (2) that there exists a threshold age that separates positive and negative contributions of reductions in mortality over time. 

% % Line interpsace
% \linespread{1}\normalsize
% 
% % Font size
% \small

%%%%%%%%%%%%%%%%%%%%%%%%%%%%%%%%%%%%%%%%%%%%%%%%%%%%%%%%%%%%%%%%%%%%%%%%%%%%%%%%%%%%%%%
%%%%%%%%%%%%%%%%%%%%%%%%%%%%%%%%%%%%%%%%%%%%%%%%%%%%%%%%%%%%%%%%%%%%%%%%%%%%%%%%%%%%%%%%
\newpage
\linespread{1}\normalsize

% References
\bibliographystyle{DemRes}
\bibliography{Bib_FormalDemo}

%%%%%%%%%%%%%%%%%%%%%%%%%%%%%%%%%%%%%%%%%%%%%%%%%%%%%%%%%%%%%%%%%%%%%%%%%%%%%%%%%%%%%%%%
%%%%%%%%%%%%%%%%%%%%%%%%%%%%%%%%%%%%%%%%%%%%%%%%%%%%%%%%%%%%%%%%%%%%%%%%%%%%%%%%%%%%%%%%
\newpage

\section*{Appendix}

% Modify enumariton of Equations in the Appendix.
\setcounter{equation}{0}
\renewcommand{\theequation}{A\arabic{equation}}

% PROPOSITION 1
\begin{theorem}
 Let $e^\dagger(x)=\int_x^\infty d(a)\,e(a)\,da\,/\,\ell(x)$ be a measure of lifespan disparity above age $x$, where $d(a)$ accounts for the distribution of deaths, $e(a)$ the remaining life expectancy at age $a$, and $\ell(x)$ is the probability of surviving from birth to age $x$. Then,
 \begin{equation}
  e^\dagger(x)=\frac{1}{\ell(x)}\int_x^\infty\,\ell(a)\big(H(a)-H(x)\big)\,da\;,
  \label{eq:edaggerAppendix}
 \end{equation}
 where $H(x)$ is the cumulative hazard to age $x$.
 \label{prop1}
\end{theorem}

\begin{proof}
  Note that
  $$
  \frac{1}{\ell(x)}\int_x^\infty\,\ell(a)\big(H(a)-H(x)\big)\,da=\frac{1}{\ell(x)}\int_x^\infty\,\ell(a)\int_x^a\mu(y)\,dy\,da\;,
  $$
  where function $\mu(y)$ is the force of mortality or hazard rate. By reversing the order of integration, and using that $e(y)=\int_y^\infty\ell(a)\,da\,/\,\ell(y)$ and $d(y)=\mu(y)\,\ell(y)$, we get
  \begin{equation*}
    \begin{split}
      \frac{1}{\ell(x)}\int_x^\infty\,\ell(a)\int_x^a\mu(y)\,dy\,da
	  & = \frac{1}{\ell(x)}\int_x^\infty\mu(y)\int_y^\infty\ell(a)\,da\,dy			\\
	  & = \frac{1}{\ell(x)}\int_x^\infty\mu(y)\,\ell(y)\,e(y)\,dy				\\			
	  & = \frac{1}{\ell(x)}\int_x^\infty d(y)\,e(y)\,dy				\\
	  & = e^\dagger(x)\;,
    \end{split}  
  \end{equation*}
  which proves~\eqref{eq:edaggerAppendix}.
\end{proof}
\medskip

% PROPOSITION 2
\begin{theorem}
  Let $\ell(x)$ be the probability of surviving from birth to age $x$. Let $\xbar{H}$ be Keyfitz' entropy and $\xbar{H}(x)=e^\dagger(x)\,/\,e(x)$ Keyfitz' entropy conditioned on reaching age $x$. Let $H(x)$ be the cumulative hazard to age $x$. Then, $g(x)=H(x)+\xbar{H}(x)-1-\xbar{H}$ is an increasing function.
 \label{prop2}
\end{theorem}

\begin{proof}
In order to demonstrate that $g(x)$ is an increasing function it is sufficient to show that its first derivative is always positive. Hence, we must prove that
\begin{equation}
 \frac{\partial}{\partial x}\,g(x)=\frac{\partial}{\partial x}\left(H(x)+\xbar{H}(x)-1-\xbar{H}\right)=\frac{\partial}{\partial x}\,H(x)+\frac{\partial}{\partial x}\xbar{H}(x)\geq0
 \label{eq:prop2}
\end{equation}
for all ages $x$. 

By the fundamental theorem of calculus,

\begin{equation}
  \frac{\partial}{\partial x}\,H(x) = \frac{\partial }{\partial x} \int_0^x\mu(a)\,da =\mu(x)\;,
  \label{Cumhaz.derv}
\end{equation}
whereas
\begin{equation*}
  \frac{\partial}{\partial x}\,\xbar{H}(x)= \frac{\partial}{\partial x}\,\left(\frac{e^\dagger(x)}{e(x)}\right)= \frac{1}{e(x)^2}\left(e(x)\,\frac{\partial}{\partial x}\,e^\dagger(x)-e^\dagger(x)\,\frac{\partial}{\partial x}\,e(x)\right)\;.
\end{equation*}

First, note that
\begin{equation}
 \begin{split}
  \frac{\partial}{\partial x}\,e^\dagger(x)	    
            & = \frac{\partial}{\partial x}\left(\frac{1}{\ell(x)}\int_x^\infty d(a)\,e(a)\,da\right)   \\
            & = \frac{1}{\ell(x)^2}\left(\ell(x)\,\frac{\partial}{\partial x}\left(\int_x^\infty d(a)\,e(a)\,da\right)-\int_x^\infty d(a)\,e(a)\,da\,\frac{\partial}{\partial x}\,\ell(x)\right)                                    \\
            & = \frac{1}{\ell(x)^2}\left(\ell(x)\,\big(-d(x)\,e(x)\big)-\int_x^\infty d(a)\,e(a)\,da\,\big(-\mu(x)\,\ell(x)\big)\right)                         \\
            & = -\frac{\mu(x)\,\ell(x)\,e(x)}{\ell(x)}+\mu(x)\,\frac{\int_x^\infty d(a)\,e(a)\,da}{\ell(x)}                                    \\
            & = \mu(x)\left(e^\dagger(x)-e(x)\right)\;.
 \end{split}
 \label{edagger.derv}
\end{equation}
On the other hand,
\begin{equation}
  \begin{split}
  \frac{\partial}{\partial x}\,e(x) 
        & = \frac{\partial }{\partial x}\left(\frac{1}{\ell(x)}\int_x^\infty\ell(a)\,da\right)    \\
        & = \frac{1}{\ell(x)^2}\,\left(\ell(x)\,\frac{\partial }{\partial x}\left(\int_x^\infty\ell(a)\,da\right)-\int_x^\infty\ell(a)\,da\,\frac{\partial }{\partial x}\,\ell(x)\right) \\
        & = \frac{1}{\ell(x)^2}\,\left(\ell(x)\,\big(-\ell(x)\big)-\int_x^\infty\ell(a)\,da\,\big(-\mu(x)\,\ell(x)\big)\right)                                  \\
        & = e(x)\,\mu(x)-1\;.
  \end{split}
  \label{ex.derv}
\end{equation}
Therefore, using~\eqref{edagger.derv} and~\eqref{ex.derv}, we get
\begin{equation}
  \begin{split}
    \frac{\partial}{\partial x}\,\xbar{H}(x)
      & = \frac{1}{e(x)^2}\,\bigg(e(x)\,\mu(x)\left(e^\dagger(x)-e(x)\right)-e^\dagger(x)\big(e(x)\,\mu(x)-1\big)\bigg)                \\
      & = \frac{1}{e(x)^2}\left(e^\dagger(x)\,e(x)\,\mu(x)-e(x)^2\,\mu(x)-e^\dagger(x)\,e(x)\,\mu(x)+e^\dagger(x)\right)                              \\
      & = \frac{e^\dagger(x)}{e(x)^2}-\mu(x)\;.      
  \end{split}
  \label{eq:Hplus.derv}
\end{equation}

Finally, replacing~\eqref{Cumhaz.derv} and~\eqref{eq:Hplus.derv} in~\eqref{eq:prop2} yields
\begin{equation*}
 \frac{\partial}{\partial x}\,g(x)=\mu(x)+\frac{e^\dagger(x)}{e(x)^2}-\mu(x)=\frac{e^\dagger(x)}{e(x)^2}\geq0\;,
\end{equation*}
for all ages $x$, which proves that $g(x)$ is an increasing function.
\end{proof}

\end{document}